%% file: main.tex
\documentclass[11pt]{article}
\usepackage{fullpage}
\usepackage[top=2cm, bottom=4.5cm, left=2cm, right=2cm]{geometry}
\usepackage{amsmath,amsthm,amsfonts,amssymb,amscd}
\usepackage{lastpage}
\usepackage{enumerate}
\usepackage{fancyhdr}
\usepackage{mathrsfs}
\usepackage{xcolor}
\usepackage{graphicx}
\usepackage{listings}
\usepackage[colorlinks=true, urlcolor=blue, linkcolor=blue, citecolor=magenta]{hyperref}
\usepackage{cite}
\usepackage{multirow}
\usepackage{comment}
\usepackage{mathpazo}		
\usepackage{todonotes}
\usepackage{xspace}

\hypersetup{%
  colorlinks=true,
  linkcolor=red,
  linkbordercolor={0 0 1}
}

\usepackage{setspace}
\newtheorem{thm}{Theorem}[section]
\newtheorem{defn}[thm]{Definition}

\newtheorem{lema}[thm]{Lemma}
\newtheorem{hypo}[thm]{Hypothesis}

\title{A Survey on Parameterized Inapproximability: {\sc $k$-Clique}, {\sc $k$-SetCover}, and More}
\author{Xuandi Ren \\ Peking University \\ \texttt{renxuandi@pku.edu.cn}}

\begin{document}

\maketitle
\abstract{
In the past a few years, many interesting inapproximability results have been obtained from the parameterized perspective. This article surveys some of such results, with a focus on {\sc $k$-Clique}, {\sc $k$-SetCover}, and other related problems.
}

\input{contents/introduction}
\input{contents/preliminaries}
\input{contents/maxcover_and_minlabel}
\input{contents/kClique}
\input{contents/kSetCover}

~\\~

\noindent{\textbf{Acknowledgements}}

~

\noindent{
I want to express my deep gratitude to Prof. Bingkai Lin, who brought me into the beautiful world of hardness of approximation, discussed with me regularly and guided me patiently. I would also like to thank my talented friends Yican Sun and Xiuhan Wang for their bright ideas and unreserved help. I really enjoy working with them.
}

\bibliographystyle{alpha}
\bibliography{ref}

\end{document}

%% file: contents/introduction.tex
\section{Introduction}

Parameterization and Approximation are two natural ways to cope with {\sf NP}-complete optimization problems. For {\sc Clique} and {\sc SetCover}, two very basic {\sf NP}-complete problems whose parameterized version {\sc $k$-Clique} and {\sc $k$-SetCover} are also complete problems of {\sf W[1]} and {\sf W[2]}, both approximation and parameterization have been studied extensively. However, the combining of parameterization and approximation remains unexplored until recent years.

In their breakthrough work, \cite{CCK+17} showed very strong inapproximability results for {\sc $k$-Clique} and {\sc $k$-SetCover} under the hypothesis {\sf Gap-ETH}. However, although maybe plausible, {\sf Gap-ETH} is such a strong hypothesis that it already gives a gap in hardness of approximation. Thus it is still of great interest to prove the same lower bound under a gap-free assumption like {\sf ETH}, {\sf W[1]$\neq$FPT} and so on. Although these years have witnessed many significant developments along this way, the inapproximability of {\sc $k$-Clique} and {\sc $k$-SetCover} under gap-free hypotheses is still far beyond settled.

This article surveys some recent results, all of which are beautiful, full of smart ideas, and involve delicate algebraic or combinatorial tools. We hope to extract the relationship between different problems, capture the essence of successful attempts, and convey the ideas inside those results to readers.

\subsection{Organization of the Survey}
This article is organized by the problems. In Section \ref{sec:pre}, we put some preliminaries, including the definition of problems and hypotheses. In Section \ref{sec:maxcover_minlabel}, we introduce {\sc MaxCover} and {\sc MinLabel}, two problems which are not only important intermediate problems in proving parameterized inapproximability, but also of great interest themselves. In Section \ref{sec:kclique} and \ref{sec:ksetCover}, we introduce recent parameterized inapproximability results of {\sc $k$-Clique} and {\sc $k$-SetCover}, respectively.

\clearpage

%% file: contents/preliminaries.tex
\section{Preliminaries}\label{sec:pre}

In this section, we first introduce some concepts in FPT approximation, then briefly describe the problems discussed in this article, and the hypotheses which the results are based on.

\subsection{FPT Approximation}
For a parameterized optimization problem, we use $n$ to denote the input size, and the parameter $k$ usually refers to the number of elements we need to pick to obtain an optimal solution. In some problems $k$ is just the optimal solution size (e.g. {\sc $k$-Clique}), while in other problems it is not (e.g. {\sc One-Sided $k$-Biclique}). By enumerating the $k$ elements in the solution, the brute-force algorithm usually runs in time $O(n^k)$.

An algorithm for a maximization (respectively, minimization) problem is called \textit{$\rho$-FPT-approximation} if it runs in $f(k)n^{O(1)}$ time for some computable function $f$, and outputs a solution of size at least $k/\rho$ (respectively, at most $k\cdot\rho$). Here $\rho$ is called \textit{approximation ratio}. If an optimization problem admits no $f(k)$-FPT-approximation for any computable function $f$, we say this problem is \textit{totally FPT inapproximable}.

Note that since computing a constant-size solution is trivial, for a maximization problem, we only care about $o(k)$-FPT-approximation and the approximation ratio is measured in terms of only $k$. However, for a minimization problem, any computable approximation ratio is non-trivial, so if totally FPT inapproximability is already established, we can also discuss approximation ratio in terms of both $k$ and $n$.

\subsection{Problems}
If the input is divided into $k$ groups, and one is only allowed to pick at most one element from each group, we say this problem is \textit{colored}, otherwise it is \textit{uncolored}. For some problems (e.g. {\sc $k$-Clique}), the two versions are equivalent, while for some other problems (e.g. {\sc $k$-Biclique}) they are not equivalent at all. We will discuss the coloring in each problem's section separately.

Now we list the problems considered in this article. There are some other problems (e.g. {\sc MaxCover}) which are used as intermediate problems in proving hardness of approximation. We put their definitions in their separate sections since they are  a bit more complicated.

\begin{itemize}
	\item {\sc 3SAT}. The input is a 3-CNF formula $\varphi$ with $m$ clauses on $n$ variables. The goal is to decide whether there is a satisfying assignment for $\varphi$.
	\item {\sc $k$-Clique}. The input is an undirected graph $G=(V,E)$ with $n$ vertices. The goal is to decide whether there is a clique of size $k$. 
	\item {\sc Densest $k$-Subgraph}. The input is an undirected graph $G=(V,E)$ with $n$ vertices, and the goal is to find the maximum number of edges induced by $k$ vertices. 
	\item {\sc $k$-SetCover}. The input is a collection of $n$ sets $\mathcal S=\{S_1,\ldots,S_n\}$ over universe $U$. The goal is to decide whether there are $k$ sets in $\mathcal S$, whose union is $U$. 
\end{itemize}

\subsection{Hypotheses}

Here we list the hypotheses which the results are based on.

{\sf W[1]$\ne$FPT} and {\sf W[2]$\ne$FPT} are arguably the most natural hypotheses in parameterized complexity, and are often used to derive FPT time lower bounds. Since {\sc $k$-Clique} and {\sc $k$-SetCover} are complete problems of {\sf W[1]} and {\sf W[2]}, respectively, we directly use their intractability results in the statements of those two hypotheses here, and omit the definition of {\sf W}-Hierarchy.

\begin{hypo}[{\sf W[1]$\ne$FPT}]
	{\sc $k$-Clique} cannot be solved in $f(k)n^{O(1)}$ time, for any computable function $f$.
\end{hypo}

\begin{hypo}[{\sf W[2]$\ne$FPT}]
	{\sc $k$-SetCover} cannot be solved in $f(k)n^{O(1)}$ time, for any computable function $f$.
\end{hypo}

Tighter time lower bounds like $n^{\Omega(k)}$ often involves the \textit{Exponential Time Hypothesis} ({\sf ETH}). 

\begin{hypo}[Exponential Time Hypothesis ({\sf ETH})\cite{IP01,IPZ01,Tov84}]
	{\sc 3SAT} cannot be solved deterministically in $2^{o(n)}$ time, where $n$ is the number of variables. Moreover, this holds even when restricted to formulae in which $m=O(n)$, and each variable appears in at most three clauses.
\end{hypo}

There are two stronger assumptions on the intractability of {\sc 3SAT}, namely, the \textit{Gap Exponential Time hypothesis} ({\sf Gap-ETH}) and \textit{Strong Exponential Time hypothesis} ({\sf SETH}). {\sf Gap-ETH} is useful in proving strong inapproximability results for many parameterized problems, while {\sf SETH} is used to show tight time lower bounds like $n^{k-o(1)}$.

\begin{hypo}[Gap Exponential Time Hypothesis ({\sf Gap-ETH}) \cite{Din16,MR16}]
	For some constant $\varepsilon>0$, there is no deterministic algorithm which runs in $2^{o(n)}$ time can, given a {\sc 3SAT} formula on $n$ variables and $m=O(n)$ clauses,  distinguish between the following two cases:
	\begin{itemize}
		\item (Completeness) the formula is satisfiable.
		\item (Soundness) any assignment violates more than $\varepsilon$ fraction of clauses.
	\end{itemize}
\end{hypo}

Note that by current state-of-the-art PCP theorem, a {\sc 3SAT} formula $\varphi$ on $n$ variables can be transformed into a constant gap {\sc 3SAT} formula $\varphi'$ on only $n \text{polylog}(n)$ variables \cite{Din07}. Therefore, assuming ETH, constant gap {\sc 3SAT} cannot be solved in $2^{o(n/\text{polylog}(n))}$ time. A big open problem is whether linear-sized PCP exists. If so, {\sf Gap-ETH} would follow from {\sf ETH}.

\begin{hypo}[Strong Exponential Time Hypothesis ({\sf SETH}) \cite{IP01, IPZ01}]
	For any $\varepsilon>0$, there is an integer $k \ge 3$ such that no algorithm can solve {\sc $k$SAT} deterministically in $2^{(1-\varepsilon)n}$ time.
\end{hypo}

{\sf Gap-ETH} and {\sf SETH} both imply {\sf ETH}. However, no formal relationship between them is known now.

There are also randomized versions of {\sf ETH}, {\sf Gap-ETH} and {\sf SETH}, which also rule out randomized algorithms running in corresponding time. We do not separately list them here. 

One last important hypothesis is the \textit{Parameterized Inapproximability  Hypothesis} ({\sf PIH}). 

\begin{hypo}[Parameterized Inapproximability  Hypothesis ({\sf PIH}) \cite{LRSZ20}]
	For some constant $\varepsilon>0$, there is no $(1+\varepsilon)$ factor FPT approximation algorithm for {\sc Colored Densest $k$-Subgraph}.
\end{hypo}

The factor $(1+\varepsilon)$ can be replaced by any constant and is not important.

Note that if for a graph the number of edges induced by $k$ vertices is only $\binom{\varepsilon k}{2}\approx \varepsilon^2 \binom{k}{2}$, it can not have a clique of size $>\varepsilon k$. Thus, {\sf PIH} implies {\sc $k$-Clique} is hard to approximate within any constant factor in FPT time. However, the reverse direction is not necessarily true (forbiddinng small clique does not imply low density of edges), and it remains an important open question that whether {\sf PIH} holds if we assume {\sc $k$-Clique} is FPT inapproximable within any constant factor.

Another remark is that {\sf PIH} can be implied from {\sf Gap-ETH}. See Appendix A of \cite{BGKM18} for a simple proof. However, deriving {\sf PIH} from a gap-free hypothesis such as {\sf ETH} is still open.

\clearpage

%% file: contents/maxcover_and_minlabel.tex
\section{{\sc MaxCover} and {\sc MinLabel}}\label{sec:maxcover_minlabel}
In this section, we introduce two intermediate problems which are important in proving hardness of {\sc $k$-Clique} and {\sc $k$-SetCover}.

The input is the same for both problems. It is a bipartite graph $G=(U \dot\cup W, E)$, such that $U$ is partitioned into $U=U_1 \dot\cup\ldots \dot\cup U_\ell$ and $W$ is partitioned into $W=W_1\dot\cup \ldots \dot\cup W_h$. We refer to $U_i$'s and $W_j$'s as \textit{left super-nodes} and \textit{right super-nodes}, respectively, and we refer to the maximum size of left super-nodes and right super-nodes as \textit{left alphabet size} and \textit{right alphabet size}, and denote them as $|\Sigma_U|$ and $|\Sigma_W|$, respectively.

We say a {\sc MaxCover} or {\sc MinLabel} instance has \textit{projection property} if for every $i \in [\ell], j \in [h]$, one of the following holds:
\begin{itemize}
	\item Every $u \in U_i$ has exactly one neighbor $w \in W_j$.
	\item There is a full bipartite graph between $U_i$ and $W_j$.
\end{itemize}

The bipartite case just means there are no restrictions between $U_i$ and $W_j$.

Another interesting property is called \textit{pseudo projection property}, which is almost the same as projection property, except the projection direction in the first case is opposite (Every $w \in W_j$ has exactly one neighbor $u \in U_i$). 

For convenience, in an instance $\Gamma$ and for a left super-node $U_i, i \in [\ell]$, we sometimes refer to the number of right super-nodes $W_j$'s, such that the edges between $U_i$ and $W_j$ do not form a full bipartite graph, as $U_i$'s \textit{degree}. Similarly define it for every right super-nodes. We call the maximum degree over all $U_i$'s (respectively, over all $W_i$'s), the \textit{left degree} (respectively, \textit{right degree}) of $\Gamma$.

A solution to {\sc MaxCover} is a subset of vertices $S \subseteq W$ formed by picking a vertex from each $W_j$ (i.e. $|S \cap W_j|=1$ for all $j \in [h]$). We say a labeling $S$ covers a left super-node $U_i$ if there exists a vertex $u_i \in U_i$ which is a common neighbor of all vertices in $S$. The goal in {\sc MaxCover} is to find a labeling that covers the maximum fraction of left super-nodes. The value of a {\sc MaxCover} instance is defined as
$$\frac{1}{\ell}\left(\max_{\text{labeling}~S}|\{i \in [\ell] |U_i\text{ is covered by }S\}|\right)$$

A solution to {\sc MinLabel} is also a subset of vertices $S \subseteq W$, but not necessarily one vertex from each $W_j$. We say a multi-labeling $S$ covers a left super-node $U_i$ if there exists a vertex $u_i \in U_i$ which has a neighbor in $S \cap W_j$ for every $j \in [h]$. The goal in {\sc MinLabel} is to find a minimum-size multi-labeling $S$ that covers all the left super-nodes. The value of a {\sc MinLabel} instance is defined as
$$\frac{1}{h}\left(\min_{\substack{
	\text{multi-labeling }S\\ \text{ which covers every } U_i}} |S|\right)$$

There is a remark on the relationship between projection property and pseudo projection property. If the degree of each left super-node $U_i$ is bounded by some constant (it is the case when reducing from {\sc 3SAT}, see Theorem \ref{thm:maxcover_under_gapeth}), then a {\sc MaxCover} instance with projection property can be reduced to a {\sc MaxCover} instance with pseudo projection property, with a constant shrinking of the gap.

\begin{thm}
	There is a reduction which, on input a {\sc MaxCover} instance  $\Gamma=(\bigcup_{i \in [\ell]} U_i,\bigcup_{i \in [h]} W_i, E)$ with projection property, and the left degree of $\Gamma$ is bounded by a constant $q$, outputs a {\sc MaxCover} instance $\Gamma'=(\bigcup_{i \in [\ell \cdot q]}  U'_i, \bigcup_{i \in [h+\ell]}W'_i, E')$ with pseudo projection property, such that
	\begin{itemize}
		\item (Completeness) If {\sc MaxCover}$(\Gamma)=1$, then {\sc MaxCover}$(\Gamma')=1$.
		\item (Soundness) If {\sc MaxCover}$(\Gamma)<1-\varepsilon$, then {\sc MaxCover}$(\Gamma')<1-\varepsilon/q$.
	\end{itemize}
	and the right degree of $\Gamma'$ is bounded by $q$.
\end{thm}
\begin{proof}
	The reduction is straightforward: for each restriction between some $U_i, i \in [\ell]$ and $W_j, j \in [h]$, build a copy of $W_j$ on the left (there are at most $\ell \cdot q$ restrictions, thus that many copies), and the new right super-nodes are just the juxtaposition of $U$ and $W$. Each left super-node is only responsible to check one restriction in $\Gamma$. The edges between the left and the right $W_i$ parts form either a bijection or a full bipartite graph, while the edges between the left and the right $U_i$ parts form either an injection from $U$ to $W$ (as they do in $\Gamma$), or a full bipartite graph, too. It's easy to see the new instance satisfies pseudo projection property, the right degree is $\le q$, and the gap is only hurt by a factor of $q$.
\end{proof}

In the following we briefly list the inapproximability results for {\sc MaxCover} and {\sc MinLabel}, which will be introduced detailedly in subsequent subsections. We use $n$ to denote $|\Sigma|$ for simplicity.

\begin{center}
\begin{tabular}{c|c|c|c|c|c}
Problem & Assumption               & Ratio & Lower Bound           & Holds Even                                                                                               & Reference                      \\ \hline
\multirow{7}{*}{\sc MaxCover}& \multirow{4}{*}{\sf Gap-ETH} & Any Constant                   & $2^{\Omega(h)}$            & \begin{tabular}[c]{@{}c@{}}$|\Sigma_U|=|\Sigma_W|=O(1)$\\ $\ell=O(h)$\\ projection property\end{tabular} & /                              \\ \cline{3-6} & 
                         & $r/\ell$                       & $n^{\Omega(r)}$     & \begin{tabular}[c]{@{}c@{}}$|\Sigma_W|=O(1)$\\ projection property\end{tabular}                          & \multirow{3}{*}{\cite{CCK+17}} \\ \cline{3-5} 
                   &      & \multirow{2}{*}{$\gamma$}                       & \multirow{2}{*}{$n^{\Omega(h)}$}     & \multirow{2}{*}{$|\Sigma_U| \le (1/\gamma)^{O(1)}$}                                                                       & \\ & & & & \\ \cline{2-6}
& {\sf W[1]$\neq$FPT}
          & $n^{-O(\frac{1}{\sqrt h})}$    & $f(h)\cdot \text{poly}(n)$ & /                                                                                                        &     \multirow{2}{*}{\cite{KL21}}                           \\ \cline{2-5}
& {\sf ETH}                      & $n^{-O(\frac{1}{h^3})}$        & $n^{\Omega(h)}$            & /                                                                                                        &                                \\ \hline  \multirow{2}{*}{\sc MinLabel}& \multirow{2}{*}{\sf Gap-ETH} & \multirow{2}{*}{$\gamma^{-1/h}$} & \multirow{2}{*}{$n^{\Omega(h)}$} & \multirow{2}{*}{$|\Sigma_U| \le (1/\gamma)^{O(1)}$}& \multirow{2}{*}{\cite{CCK+17}} \\ 
& &  & & & \\ \hline
\end{tabular}
\end{center}

It's worth noting that \cite{KLM19} also showed some inapproximability results for \textsc{MaxCover}. However, as their parameters are specifically designed for later use of proving hardness of {\sc $k$-SetCover}, we defer their results to Section \ref{DPCPF} instead of here. 
	
\subsection{Hardness Results Based on  {\sf Gap-ETH}}
Results in this subsection are from \cite{CCK+17}.

A {\sc 3SAT} instance can be formulated as {\sc MaxCover} instance as follows. Each left super-node consists of 7 vertices, which represent the satisfying assignments of a clause. Each right super-node consists of 2 vertices, which correspond to the true/false assignment for a variable. Two vertices are linked if and only if the assignments are consistent. Therefore, {\sf Gap-ETH} can be also restated as an intractability result of constant gap {\sc MaxCover}:

\begin{thm}\label{thm:maxcover_under_gapeth}
	Assuming {\sf Gap-ETH}, there is a constant $\varepsilon>0$ such that no deterministic algorithm can distinguish between the following cases for an instance $\Gamma =(\cup_{i\in[\ell]} U_i,\cup_{i \in [h]} W_i, E)$ in $2^{o(h)}$ time:
	\begin{itemize}
		\item (Completeness) {\sc MaxCover$(\Gamma)=1$}.
		\item (Soundness) {\sc MaxCover$(\Gamma)<1-\varepsilon$}.
	\end{itemize}
	Moreover, this holds even when $|\Sigma_U|,|\Sigma_W|=O(1), \ell=\Theta(h)$ and $\Gamma$ has projection property.
\end{thm}

Actually {\sf Gap-ETH} is equivalent to the above, see Appendix E of \cite{CCK+17}.

Note that {\sc MaxCover} can be solved in $O^*(|\Sigma_U|^{\ell})$ or $O^*(|\Sigma_W|^h)$ time, by just enumerating vertices picked from each left super-nodes or right super-nodes. Moreover, it can even decide whether the answer is $\ge \frac{r}{\ell}$ in $O^*(\binom{\ell}{r}|\Sigma_U|^r)=O^*((\ell |\Sigma_U|)^r)$ or $O^*(|\Sigma_W|^h)$ time. As shown below, these are the best possible assuming {\sf Gap-ETH}.
\begin{thm}[Theorem 4.2 in \cite{CCK+17}]\label{MaxCover_L}
	Assuming {\sf Gap-ETH}, there exists constants $\delta,\rho>0$ such that for any $\ell \ge r \ge \rho$, no algorithm can take a {\sc MaxCover} instance $\Gamma$ with $\ell$ left super-nodes, distinguish between the following cases in $O_{\ell,r}(|\Gamma|^{\delta r})$ time:
	\begin{itemize}
		\item (Completeness) {\sc MaxCover$(\Gamma)=1$}.
		\item (Soundness) {\sc MaxCover$(\Gamma)<\frac{r}{\ell}$}.
	\end{itemize}
	This holds even when $|\Sigma_W|=O(1)$ and $\Gamma$ has projection property.
\end{thm}

The theorem is straightforward when $r=\Theta(\ell)$, because we can directly compress a constant number of left super-nodes into one. The interesting case is when $r$ is much smaller than $\ell$.

The proof involves a combinatorial object called \textit{disperser}, which is defined as follows.
\begin{defn}[Disperser\cite{CW89, Zuc96a, Zuc96b}]
	For positive integers $m,\ell, k, r\in \mathbb N$ and constant $\varepsilon\in(0,1)$, an $(m,\ell,k,r,\varepsilon)$-disperser is a collection $\mathcal I$ of $\ell$ subsets $I_1,\ldots,I_\ell \subseteq [m]$, each of size $k$, such that the union of any $r$ different subsets from the collection has size at least $(1-\varepsilon)m$.
\end{defn}

Dispersers with proper parameters can be constructed using random subsets with high probability.

\begin{lema}
	For positive integers $m,\ell,r \in \mathbb N$ and constant $\varepsilon \in (0,1)$, let $k=\lceil \frac{3m}{\varepsilon r}\rceil$ and let $I_1,\ldots,I_\ell$ be random $k$-subsets of $[m]$. If $\ln \ell \le \frac{m}{r}$ then $\mathcal I=\{I_1,\ldots,I_\ell\}$ is an $(m,\ell,k,r,\varepsilon)$-disperser with probability at least $1-e^{-m}$.
\end{lema}

This above construction can also be derandomized easily. With this tool, we can compress the left super-nodes in a {\sc MaxCover} instance according to those subsets. Each left super-node now corresponds to satisfying assignments of the AND of $k$ clauses. If there is a labeling that covers at least $r$ left super-nodes, from the definition of disperser we know that at least $(1-\varepsilon) m$ clauses in the original {\sc 3SAT} instance can be simultaneously satisfied. The size of the new instance is $2^{O(k)}=2^{O(m/r)}$, thus an algorithm for {\sc MaxCover} which runs in $|\Gamma|^{o(r)}$ time would lead to an algorithm for constant gap {\sc 3SAT} in $2^{o(m)}$ time, refuting {\sf Gap-ETH}. 

In the other direction, we would like to rule out $|\Gamma|^{o(h)}$ algorithms for approximating {\sc MaxCover}, where $h$ is the number of right super-nodes. We have the following theorem:

\begin{thm}[Theorem 4.3 in \cite{CCK+17}] \label{MaxCover_R}
Assuming {\sf Gap-ETH}, there exists constants $\delta,\rho>0$ such that for any $h \ge \rho$ and $1 \ge \gamma >0$, no algorithm can take a {\sc MaxCover} instance $\Gamma$ with $h$ right super-nodes, distinguish between the following cases in $O_{h,\gamma}(|\Gamma|^{\delta h})$ time:
\begin{itemize}
	\item (Completeness) {\sc MaxCover$(\Gamma)=1$}.
	\item (Soundness) {\sc MaxCover$(\Gamma)<\gamma$}.
\end{itemize}
This holds even when $|\Sigma_U|\le (1/\gamma)^{O(1)}$.
\end{thm}

Note that in the statement of this theorem, $h$ can be any fixed constant, while in the statement of {\sf Gap-ETH}, $h$ is the number of variables which goes to infinity. Thus, a straightforward idea is to compress the variables into $h$ groups, each of size $n/h$.

After grouping variables, we can also compress the clauses in order to amplify the soundness parameter to $\gamma$. Specifically, let $k=\ln(\frac{1}{\gamma})/\varepsilon$, take $\ell=\binom{m}{k}$ left super-nodes, each corresponding to satisfying assignments of the AND of some $k$ clauses. If only $(1-\varepsilon) m$ clauses in the original {\sc 3SAT} instance can be satisfied, then only $\binom{(1-\varepsilon)m}{k}$ clauses in the new instance can be satisfied, leading to a soundness parameter $\binom{(1-\varepsilon)m}{k}/\binom{m}{k}\le e^{-\varepsilon k}=\gamma$. Furthermore, $|\Sigma_U|=O(1)^{k}\le (1/\gamma)^{O(1)}$.

One important thing is that we need to make sure $\ell$ and $|\Sigma_U|$ can be bounded by $|\Sigma_W|=2^{n/h}$, so that $|\Sigma_W|$ is the dominating term in $|\Gamma|$. Thus, $\gamma$ cannot be arbitrarily small. Fortunately in its major applications (e.g. hardness of {\sc SetCover}), this will not be the bottleneck.

Next we proceed to discuss hardness of {\sc MinLabel}.

\begin{thm}[Theorem 4.4 in \cite{CCK+17} ]\label{MinLabel}
	Assuming {\sf Gap-ETH}, there exists constants $\delta,\rho>0$ such that, for any $h \ge \rho$ and $1 \ge \gamma > 0$, no algorithm can take a {\sc MinLabel} instance $\Gamma$ with $h$ right super-nodes, distinguish between the following cases in $O_{h,\gamma}(|\Gamma|^{\delta h})$ time:
	\begin{itemize}
		\item (Completeness) {\sc MinLabel$(\Gamma)=1$}.
		\item (Soundness) {\sc MinLabel$(\Gamma)>\gamma^{-1/h}$}.
	\end{itemize}
	This holds even when $|\Sigma_U| \le (1/\gamma)^{O(1)}$.
\end{thm}
The instance is exactly the one in Theorem \ref{MaxCover_R}. It is easy to see that in the completeness case, {\sc MaxCover}$=1$ implies {\sc MinLabel}$=1$. We only need to additionally argue that, in the soundness case, small {\sc MaxCover} implies large {\sc MinLabel}. Prove by contradiction, if {\sc MinLabel}$\le\gamma^{-1/h}$, we can fix this multi-labeling, and pick a vertex uniformly at random from each right super-node to form a labeling. By proving the expected fraction of covered left super-nodes $\ge \gamma$, we have {\sc MaxCover}$\ge \gamma$. In the following we suppose the optimal right multi-labeling $S \subseteq W$ covers left vertices $u_1 \in U_1,\ldots,u_h \in U_h$.
$$\begin{aligned}
	{\sc MaxCover} \ge & \mathbb E\left[\frac{1}{\ell}\sum_{i=1}^{\ell}[u_i \text{ is covered}]\right] \\ = & \frac{1}{\ell}\sum_{i=1}^{\ell} \prod_{j=1}^{h}|\mathcal N(u_i,W_j \cap S)|^{-1}\\
	\ge & \frac{1}{\ell}\sum_{i=1}^{\ell} \left( \frac{1}{h} \sum_{j=1}^h|\mathcal N(u_i,W_j \cap S)|\right)^{-h} \\
	\ge & \frac{1}{\ell}\sum_{i=1}^{\ell}\left(\frac{|S|}{h}\right)^{-h} \\
	= & \gamma
\end{aligned}$$

The left alphabet size is $(1/\gamma)^{O(1)}$, as in Theorem \ref{MaxCover_R}.

\subsection{Gap Producing via Threshold Graph Composition}
\label{TGC-KN21}
\textit{Threshold Graph Composition} is a powerful gap-producing technique. It was first proposed by Lin in his breakthrough work \cite{Lin15}, and has been used to create gap for many parameterized problems \cite{CL16,Lin19,BBE+19,KL21}. 

At a high level, in TGC we compose an instance which has no gap, with a threshold graph which is oblivious to the instance, to produce a gap instance of that problem. The two main challenges of TGC are the creation of a threshold graph with desired properties, and the right way to compose the input and the threshold graph, respectively. 

In this subsection, we introduce a delicate threshold graph, which is constructed via error correcting codes. This graph was proposed by Karthik et al. \cite{KL21}, and had many applications such as in proving hardness of {\sc MaxCover} starting from {\sf W[1]$\neq$FPT} or {\sf ETH} (later in this Section), and in simplifying the proof of {\sc $k$-SetCover} inapproximability in \cite{Lin19} (see Section \ref{SetCover-TGC}).

We first formalize some definitions related to error correcting codes.

\begin{defn}[Error Correcting Codes]
	Let $\Sigma$ be a finite set, for every $\ell \in \mathbb N$ a subset $C:\Sigma^r \to \Sigma^\ell$ is an error correcting code with message length $r$, block length $\ell$ and relative distance $\delta$ if for every $x,y \in \Sigma^r$, $\Delta(C(x),C(y))\ge \delta$.  We denote then $\Delta(C)=\delta$. Here $\Delta(x,y)=\frac{1}{\ell} |\{i \in [\ell] |x_i \ne y_i\}|$.
\end{defn}

We sometimes abuse notations and treat an error correcting code as its image, i.e., $C \subset \Sigma^{\ell}$.

\begin{defn}[Collision Number]
	The collision number of an error correcting code $C$ is the smallest number $s$ such that there exists a set $S \subseteq C$ with $|S|=s$, and for every $j \in [\ell]$, there are two strings $x,y \in S$ such that $x_j=y_j$. We denote this number as $Col(C)$.
\end{defn}

For any error correcting code $C:\Sigma^r \to \Sigma^\ell$ and any $k \in \mathbb N$, we can take it to build a bipartite threshold graph $G=(A \dot\cup B, E)$ with the following properties efficiently:
\begin{itemize}
	\item $A$ is divided into $k$ groups, each of size $|\Sigma|^r$. $B$ is divided into $\ell$ groups, each of size $|\Sigma|^k$.
	\item (Completeness) For any $a_1 \in A_1,\ldots a_k \in A_k$ and for any $j \in [\ell]$, there is a unique vertex $b \in B_j$ which is a common neighbor of $\{a_1,\ldots,a_k\}$.
	\item (Soundness) For every $i \in [k]$ and every distinct $a \ne a'\in A_i$, for $\Delta(C)\cdot \ell$  of the parts $j \in [\ell]$, we have that $\mathcal N(a) \cap \mathcal N(a') \cap \mathcal B_j=\emptyset$.
	\item (Collision Property) Let $X \subseteq A$ such that for every $j \in [\ell]$, there exists $b \in B_j$ which is a common neighbor of (at least) $k+1$ vertices in $X$. Then $|X| \ge Col(C)$.
\end{itemize}

The graph is constructed as follows:

\begin{itemize}
	\item For every $i \in [k]$, we associate $A_i$ with all codewords in $C$, i.e. each vertex in $A_i$ is a unique codeword in the image of $C$. 
	\item For every $j \in [\ell]$, we associate $B_j$ with the set $\Sigma^k$.
	\item A vertex $a \in A_i$ and a vertex $b \in B_j$ are linked if and only if $a_j=b_i$.
\end{itemize}

We can think of the graph as an $k \times \ell$ matrix, when picking vertices from each $A_i, i \in [k]$, we are filling the codewords into each row of the matrix. By reading out each column of the matrix, we can pick exactly one common neighbor of them in each $B_j, j \in [\ell]$, satisfying the completeness property.

The soundness property is also straightforward: if we pick two vertices $a,a'$ from the same left group $A_i$, the two codewords differ in at least $\Delta(C) \cdot \ell$ positions. For those columns we cannot ``read out'' any $k$-bit string whose $i$-th bit equals to two different characters.

As for the collision property, for a set $X \subseteq A$, if for every $j \in [\ell]$ there exists $b \in B_j$ which is a common neighbor of at least $k+1$ vertices in $X$, it's easy to see that for any $j \in [\ell]$, we can pick $x,y \in X$ such that $x_j=y_j$ by pigeonhole principle.

Now we describe how to compose this threshold graph with a {\sc $k$-MaxCover} instance $\Gamma$ where the parameter $k$ denotes the number of right super-nodes, to produce a {\sc Gap $k$-MaxCover} instance $\Gamma'$ such that:
\begin{itemize}
	\item (Completeness) If {\sc MaxCover}$(\Gamma)=1$, then {\sc MaxCover}$(\Gamma')=1$.
	\item (Soundness) If {\sc MaxCover}$(\Gamma)<1$, then {\sc MaxCover}$(\Gamma')\le 1-\Delta(C)$.
\end{itemize}

Given a {\sc $k$-MaxCover} instance $\Gamma=G=(U \dot\cup W, E)$ with pseudo projection property, where $W=W_1\dot\cup\ldots\dot\cup W_k$ and $U=U_1\dot\cup\ldots\dot\cup U_t$, and a threshold graph $G'=(A \dot\cup B,E')$, where $A=A_1\dot\cup\ldots\dot\cup A_t$ and $B=B_1 \dot\cup\ldots\dot\cup B_\ell$, w.l.o.g. assume $|\Sigma|^r \ge \max_{i=1}^t |U_i|$, we build the new {\sc $k$-MaxCover} instance $\Gamma'$ as follows.
\begin{itemize}
	\item Arbitrarily match every vertex $u_i \in U_i$ to a vertex in $a_i \in A_i$ without repetitions. This can be done since $|\Sigma|^r \ge \max_{i=1}^t |U_i|$.
	\item The new right super-nodes are $W_1 \ldots W_k$, and the new left super-nodes are $B_1 \ldots B_\ell$.
	\item A right vertex $w \in W_i$ and a left vertex $b \in B_j$ are linked if and only if there exists $u_1 \in U_1,\ldots, u_t \in U_t$ such that $w_i$ is linked to each $u_1 \ldots u_t$ in $G$, and $b$ is linked to the matching $a_1 \ldots a_t$ in $G'$.
\end{itemize}

The completeness case is obvious, by picking one $w$ in each right super-node $W_i$, there is a common neighbor in each left super-node $U_i$. Consider their matching vertices $a_1 \ldots a_t$, there is exactly one common neighbor in each $B_i$. 

The soundness case needs the pseudo projection property of $G$, i.e., for every $i \in [k]$ and $j \in t$, edges between $W_i$ and $U_j$ either form a function, or are complete. Fix any labeling $w_1 \in W_1, \ldots, w_k \in W_k$, there must be a super-node $U_j$ which cannot be covered. This means there must be two left vertices $w,w'$ mapping to different vertices in $U_j$. Let the two vertices be $u,u' \in U_j$, and let the matching vertices of them be $a,a' \in A_j$, by the soundness property of threshold graph, only in $1-\Delta(C)$ fraction of parts $j \in [\ell]$ is there a common neighbor of $a,a'$ in $B_j$. This means the labeling $\{w_1,\ldots,w_k\}$ can only cover $(1-\Delta(C))\cdot \ell$ parts of $B$, i.e., {\sc MaxCover}$(\Gamma')\le 1-\Delta(C)$.

Next we use this technique to prove strong inapproximability results of {\sc $k$-MaxCover} based on {\sf W[1]$\neq$FPT} and {\sf ETH}.

\begin{thm}[Theorem 4.3 in \cite{KL21} ]
	Assuming {\sf W[1]$\neq$FPT}, for any computable function $f$, there is no $f(k)\cdot \text{poly}(n)$ time algorithm which can take a {\sc MaxCover} instance $\Gamma$ with $k$ right super-nodes, distinguish between the following two cases:
	\begin{itemize}
		\item (Completeness) {\sc MaxCover$(\Gamma)=1$}.
		\item (Soundness) {\sc MaxCover$(\Gamma)\le n^{-O(\frac{1}{\sqrt k})}$}.
	\end{itemize}
\end{thm}

\begin{proof}
	First reduce {\sc $k$-Clique} to {\sc MaxCover} with $K=\binom{k}{2}$ right super-nodes and $k$ left super-nodes in the canonical way. Note that this {\sc MaxCover} instance has pseudo projection property. Then take a Reed-Solomon Code to build the threshold graph. To ensure the right alphabet size (which is $|\Sigma|^k$ here) $\le n$, we need $|\Sigma| \le n^{\frac{1}{k}}$, and to ensure $|\Sigma|^r \ge \max_{i=1}^t |U_i|=n$, we need $r\ge\log_{|\Sigma|} n =k$. Thus according to the properties of Reed-Solomon Code, the soundness parameter is $1-\Delta(C)=1-(1-\frac{r}{|\Sigma|})=n^{-O(\frac{1}{k})}=n^{-O(\frac{1}{\sqrt K})}$.
\end{proof}

\begin{thm}[Theorem 4.4 in \cite{KL21} ]
	Assuming {\sf ETH}, there is no $n^{o(k)}$ time algorithm which can take a {\sc MaxCover} instance $\Gamma$ with $k$ right super-nodes, distinguish between the following two cases:
	\begin{itemize}
		\item (Completeness) {\sc MaxCover$(\Gamma)=1$}.
		\item (Soundness) {\sc MaxCover$(\Gamma) \le n^{-O(\frac{1}{k^3})}$}.
	\end{itemize}
\end{thm}

\begin{proof}
	First reduce {\sc 3SAT} to {\sc MaxCover} with $k$ right super-nodes and $t=\binom{k}{1}+\binom{k}{2}+\binom{k}{3}$ left super-nodes. Each right super-node corresponds to satisfying assignments of some $m/k$ clauses, and thus has $N=2^{\Theta(n/k)}$ vertices in it. Each left super-node corresponds to assignments to variables which appear in exactly some one/two/three groups of clauses. This {\sc MaxCover} instance also has pseudo projection property. Note that here it's necessary to group the variables to change the number of left super-nodes from $n$ to $\binom{k}{1}+\binom{k}{2}+\binom{k}{3}$, because in our construction of threshold graph, the size of each new right super-node $B$ is $|\Sigma|^t$, which is too large if $t=n$. After that, we still use Reed-Solomon Codes. Now to ensure $|\Sigma|^t \le N$ we need $|\Sigma| \le 2^{n/k^4}$, and to ensure $|\Sigma|^r \ge N$ we need $r \ge \log_{|\Sigma|}N=\Omega(k^3)$. The soundness parameter is $1-\Delta(C)=1-(1-\frac{r}{|\Sigma|})=\frac{k^3}{2^{n/k^4}}=N^{-O(\frac{1}{k^3})}$.
\end{proof}

\clearpage

%% file: contents/kClique.tex
\section{\sc $k$-Clique}
\label{sec:kclique}

{\sc Clique} is arguably the first natural combinatorial optimization problem. Its inapproximability in the NP regime is studied extensively \cite{BGLR93,BS94,Has96,FGL+96,Gol98,FK00,Zuc07}. However, in the parameterized perspective, although {\sc $k$-Clique} is known to be the complete problem of {\sf W[1]}, and even cannot be solved in $n^{o(k)}$ time assuming {\sf ETH} \cite{CHKX06b}, there is still a lot of work to do on its parameterized inapproximability. 

To approximate {\sc $k$-Clique} to a factor of $\rho$, we only need to compute a clique of size $k/\rho$, which can be trivially done in $n^{k/\rho}$ time. In their milestone work, Chalermsook et al.  \cite{CCK+17} showed this cannot be improved assuming {\sf Gap-ETH}. However, results based on non-gap assumptions are not reached until very recently Lin \cite{Lin21} showed constant approximating {\sc $k$-Clique} is {\sf W[1]}-hard. He also obtained an $n^{\Omega(\sqrt[5]{\log k})}$ lower bound for constant approximating {\sc $k$-Clique} based on {\sf ETH}, and this bound was recently improved to $n^{\Omega(\log k)}$ by \cite{LRSW21}.

The following table lists current state-of-the-art inapproximability results of {\sc $k$-Clique} based on different hypotheses. Here $f$ can be any computable function. 

\begin{center}
\begin{tabular}{c|c|c|c}
Complexity Assumption & Inapproximability Ratio & Time Lower Bound   & Reference   \\ \hline
\multirow{2}{*}{\sf W[1]$\neq$FPT} & \multirow{2}{*}{Any constant} & \multirow{2}{*}{$f(k)\cdot \text{poly}(n)$} & \multirow{2}{*}{\cite{Lin21}} \\ & & & \\ \hline
\multirow{2}{*}{\sf PIH} & \multirow{2}{*}{Any constant} & \multirow{2}{*}{$f(k)\cdot \text{poly}(n)$} & \multirow{2}{*}{/} \\ & & & \\ \hline
\multirow{2}{*}{{\sf ETH}} & Any constant & $f(k)\cdot n^{\Omega(\log k)}$ & \multirow{2}{*}{\cite{LRSW21}} \\
  & $k^{o(1)}$ & $f(k)\cdot \text{poly}(n)$ & \\ \hline
\multirow{2}{*}{\sf Gap-ETH} & \multirow{2}{*}{$\rho=o(k)$} & \multirow{2}{*}{$f(k) \cdot n^{\Omega(k/\rho)}$} & \multirow{2}{*}{\cite{CCK+17}} \\ & & & \\ \hline
\end{tabular}
\end{center}

We shall notice that the colored version and uncolored version of {\sc $k$-Clique} are equivalent, because a colored version can be interpreted as an uncolored version by leaving each group an independent set, and we can make $k$ different copies of the original graph to transform an uncolored version to a colored version.

\subsection{Reduction from {\sc MaxCover} with Projection Property}

The {\sf Gap-ETH} hardness of {\sc $k$-Clique} directly follows from Theorem \ref{MaxCover_L}. Since the instance has projection property, two left vertices agree if and only if they map to the same vertex in each right super-node, and $k$ left vertices agree if and only if they pairwise agree. Thus, we can transform a {\sc MaxCover} instance in Theorem \ref{MaxCover_L} with $k$ left super-nodes to an {\sc $k$-Clique} instance with the same value.

Therefore, as pointed out in Theorem \ref{MaxCover_L}, even deciding if there is a clique of size $\ge r$ needs $n^{\Omega(r)}$ time. Here $r$ can be any $\omega(1)$, which means approximating {\sc $k$-Clique} to any $\rho=k/r=o(k)$ ratio cannot be done in $f(k) \cdot n^{o(r)}=f(k)\cdot n^{o(k/\rho)}$ time.

Note the projection property is crucial, without which a {\sc MaxCover} instance cannot be reduced to an {\sc $k$-Clique} instance, because the agreement test of $k$ left vertices cannot be decomposed locally to agreement tests of $\binom{k}{2}$ pairs of left vertices.

One interesting thing is that the optimal inapproximability result of {\sc $k$-Clique} can be obtained from hypotheses other than (but similar to) {\sf Gap-ETH}, see the following as an example.

\begin{thm}
	Given an undirected graph with $n$ groups of $O(1)$ vertices, each forming an independent set, if there exists a constant $\varepsilon>0$ such that distinguishing between the following cases cannot be done in $2^{o(n)}$ time:
	\begin{itemize}
		\item (Completeness) there is a clique of size $n$.
		\item (Soundness) there is no clique of size $\varepsilon n$.
	\end{itemize}
	then {\sc $k$-Clique} cannot be approximated to any $\rho=o(k)$ ratio in $f(k)\cdot n^{o(k/\rho)}$ time.
\end{thm}

This assumption is a little weaker than {\sf Gap-ETH} since it can be obtained through the canonical reduction from {\sc 3SAT} to {\sc Clique}.

The proof is almost the same as that in Theorem \ref{MaxCover_L}: just compose the groups using a disperser. Details are omitted here.

\subsection{\sc $k$-VectorSum}

Before introducing {\sf W}[1]-hardness of constant approximating {\sc $k$-Clique}, we want to mention an important {\sf W[1]}-complete problem, {\sc $k$-VectorSum}, which is used as an intermediate problem in the reduction in \cite{Lin21}. 

In the {\sc $k$-VectorSum} problem, we are given $k$ groups of vectors $V_1,\ldots, V_k \subseteq \mathbb F^d$ together with a target vector $\vec t \in \mathbb F^d$, where $\mathbb F$ is some finite field and $d$ is the dimension of vectors. The goal is to decide whether there exists vectors $\vec v_1 \in V_1,\ldots, \vec v_k \in V_k$ such that $\sum_{i=1}^k \vec v_i=\vec t$.

It's easy to see {\sc $k$-VectorSum} with $|\mathbb F|=O(1)$ and $d=O(k\log n)$ is {\sf W[1]}-hard, and even does not admit $n^{o(k)}$ time algorithms assuming {\sf ETH}. The idea is to use entries of vectors to check the consistency in {\sc $k$-Clique} or {\sc 3SAT}.

\begin{thm}\label{thm:kclique_to_kvectorsum}
	Assuming {\sf W[1]$\neq$FPT}, {\sc $k$-VectorSum} with $\mathbb F=\mathbb F_2$ and $d=\Theta(k\log n)$ can not be solved in $f(k)\cdot \text{poly}(n)$ time.
\end{thm}
\begin{proof}
	Set $K=\binom{k}{2}$ groups of vectors, each representing valid edges between an $i$-th block and a $j$-th block of vertices in {\sc $k$-Clique}. For each $i \in [k]$, we want to make sure that, the $(k-1)$ vertices chosen from the $i$-th block are all the same one. Thus we need to do $k \cdot (k-2)$ equality checks, each on two $(\log n)$-bit strings. In short, we exploit a new entry to do each bitwise equality check: set the entry to be the unchecked bit in the two vectors involved, and set the entry to be zero in all other vectors. Let the target vector to be $\vec {\mathbf 0}$. Thus all vectors sum up to zero in this entry if and only if the two to-be-checked bits are the same. The produced {\sc $K$-VectorSum} instance has parameter $K=\Theta(k^2)$ and dimension $d=\Theta(k^2 \log n)$.
\end{proof}

\begin{thm}
	Assuming {\sf ETH}, {\sc $k$-VectorSum} with $\mathbb F=\mathbb F_2$ and $d=\Theta(k \log n)$ can not be solved in $n^{o(k)}$ time.
\end{thm}
\begin{proof}
	Divide the clauses into $k$ equal-sized groups. Enumerate satisfying partial assignments of each group of clauses, then we want to check the consistency of those partial assignments. Note we can assume that each variable only appears in at most 3 clauses, thus we only need to do at most 2 pair-wise equality checks (between the first and the second appearances, and between the second and the third). The equality check step is the same as that in the proof of Theorem \ref{thm:kclique_to_kvectorsum}, and is omitted here.
	
	The produced {\sc $k$-VectorSum} instance has size $N=2^{\Theta(n/k)}$ and the dimension is $d=O(n)=O(k \log N)$. Assuming {\sf ETH}, this can not be solved in $N^{o(k)}$ time.
\end{proof}
Note that the $\mathbb F_2$ can be replaced by any finite field, as long as we slightly adjust the value of an entry: both 0 if $x=0$; $c$ in the first appearance and $-c$ in the second appearance for some constant $c\ne 0$ if $x=1$.

\subsection{Gap Producing via Hadamard Codes}
In this subsection we briefly introduce how Lin \cite{Lin21} rules out constant FPT-approximations of {\sc $k$-Clique} under {\sf W[1]$\neq$FPT}.

The most essential step in the reduction is to combine {\sc $k$-VectorSum}, whose {\sf W[1]}-hardness was shown in Theorem \ref{thm:kclique_to_kvectorsum}, with Hadamard codes to create a gap. The technique is very similar to which people used in proving weakened PCP theorem, namely, in proving {\sf NP} $\subseteq$ {\sf PCP}$(\text{poly(n),1})$ \cite{ALM+98}.

Here follows the definition of Hadamard Code, and its two important properties.

\begin{defn}[Walsh-Hadamard Code]
	For two strings $x,y \in \{0,1\}^n$, define $\langle x,y\rangle=\sum_{i=1}^n x_iy_i \pmod 2$. The Walsh-Hadamard Code is the function $f:\{0,1\}^n \to \{0,1\}^{2^n}$ that maps every string $x \in \{0,1\}^n$ into the string $z \in \{0,1\}^{2^n}$ such that $z_y = \langle x,y \rangle$ for every $y \in \{0,1\}^n$.
\end{defn}

\begin{enumerate}
	\item (Linearity Testing, \cite{BLR93}) Let $g$ be a function mapping from $\{0,1\}^n$ to $\{0,1\}$. If it can pass $(1-\delta)$ fraction of tests $g(x)+g(x') = g(x+x'), x,x' \in \{0,1\}^n$, then $g$ is at least $(1-\delta)$-close to a true linear function $\tilde g:\{0,1\}^n \to \{0,1\}$, which can also be parsed as a Hadamard codeword. By setting $0 \le \delta <\frac{1}{4}$, this codeword is unique since the relative distance of Hadamard Code is exactly $\frac{1}{2}$.
	\item (Locally Decodable Property) Suppose $g$ is $(1-\delta)$-close to a Hadamard codeword $f(x)$ for some $x \in \{0,1\}^n$, then for any $y \in \{0,1\}^n$, we can recover $(f(x))_y$ probabilistically by querying only two positions of $g$. Namely, sample $y' \in \{0,1\}^n$ uniformly at random, and output $g(y+y')+g(y')$. This succeeds with probability at least $1-2\delta$ by a simple union bound.
\end{enumerate}

Given an {\sc $k$-VectorSum} instance $(V_1,\ldots,V_k,\vec t)$, we build a {\sc CSP} instance on variable set $X=\{x_{\vec a_1,\ldots,\vec a_k} : \vec a_1,\ldots,\vec a_k\in\mathbb{F}^d\}$. Let $\vec{v}_1 \in V_1,\ldots,\vec{v}_k \in V_k$ be a solution that sums up to $\vec t$ in the yes case, each variable $x_{\vec a_1,\ldots,\vec a_k}$ is supposed to take the value $\sum_{i \in [k]} \langle \vec a_i, \vec v_i\rangle$. Here we are actually concatenating the $k$ solution vectors into one long vector $\vec v \in \mathbb F^{kd}$, and the concatenation of variables is supposed to be the Hadamard codeword of $\vec v$.

There are three types of tests we want to make:

\begin{itemize}
	\item (T1) $\forall \vec a_1,\ldots,\vec a_k,\vec b_1,\ldots\vec b_k \in \mathbb F^d$, test whether $x_{\vec a_1,\ldots,\vec a_k}+x_{\vec b_1,\ldots\vec b_k} = x_{\vec a_1 + \vec b_1, \ldots, \vec a_k + \vec b_k}$.
	\item (T2) $\forall i \in [k],\forall \vec a_1,\ldots,\vec a_k,\vec a \in \mathbb F^d$, test whether $x_{\vec a_1,\ldots,\vec a_i+\vec a,\ldots,\vec a_k}-x_{\vec a_1,\ldots,\vec a_k}=\langle \vec a,\vec v\rangle$ for some $v \in V_i$.
	\item (T3) $\forall \vec a_1,\ldots,\vec a_k,\vec a \in \mathbb F^d$, test whether $x_{\vec a_1+\vec a,\ldots,\vec a_k+\vec a}-x_{\vec a_1,\ldots,\vec a_k}=\langle \vec a,\vec t \rangle$.
\end{itemize}

By linearity testing, if an assignment to $X$ satisfies $(1-\delta)$-fraction of constraints in (T1), then $X$ is at least $(1-\delta)$-close to a true Hadamard codeword $f(\vec u)$, $\vec u \in \mathbb F^{kd}$. The constraints (T2) and (T3) use locally decodable properties to recover $f(\vec u)_{(\vec 0,\ldots,\vec a,\ldots,\vec 0)}, f(\vec u)_{(\vec a,\ldots,\vec a)}$, and use them to check whether $\vec u$ indeed indicates a satisfying solution of our {\sc $k$-VectorSum} instance.

After that, we use a slightly modified FGLSS reduction to build an $2^{\text{poly}(k)\cdot d}${\sc-clique} instance. We build a group for each variable indicating its possible values, and a group for each (T1) test. A variable vertex and a test vertex are linked either if they are consistent, or the test is irrelevant to the variable. Two test vertices are linked in a similar way accordingly. Two variable vertices are linked either if the values specified by them pass the (T2) and (T3) tests, or there are no such tests between them. Therefore, if there is a clique whose size is at least $(1-\varepsilon)$ times the maximum, the following conditions hold:
\begin{enumerate}
	\item A constant fraction of (T1) tests are passed.
	\item For a constant fraction of variables, all (T2) and (T3) tests between them are passed.
\end{enumerate}

This completes the reduction. 

There are still some technical details. For example, in this reduction the number of vertices is $2^{\text{poly}(k)\cdot d}=n^{\text{poly}(k)}$, which is too large. \cite{Lin21} sampled some random matrices to handle this. Details are omitted here.

\clearpage

%% file: contents/kSetCover.tex
\section{\sc $k$-SetCover}
\label{sec:ksetCover}

{\sc SetCover}, which is equivalent to the {\sc Dominating Set} problem, is one of the fundamental problems in computational complexity. A simple greedy algorithm yield a $(\ln n)$-approximation of this problem. On the opposite side, it was shown that $(1-\varepsilon)\ln n$-approximation for this problem is {\sf NP}-hard for every $\varepsilon>0$ \cite{DS14}. Thus, its approximability in the {\sf NP} regime has been completely settled.

In the parameterized regime, 
{\sc $k$-SetCover} is the complete problem of {\sf W[2]}, and does not admit $n^{o(k)}$ algorithms assuming {\sf ETH} \cite{CHKX06b}, not even $n^{k-\varepsilon}$ algorithms assuming {\sf SETH} \cite{PW10}. Hardness of approximation of {\sc $k$-SetCover} in FPT time was studied by \cite{CL16,CCK+17,KLM19,Lin19}, and currently based on {\sf Gap-ETH}, {\sf ETH}, {\sf W[1]$\neq$FPT} or {\sf $k$-SUM Hypothesis}, {\sc $k$-SetCover} is hard to approximate to within a $(\log n)^{\frac{1}{\text{poly}(k)}}$ factor. In one direction, we wonder if this $(\log n)^{\frac{1}{\text{poly}(k)}}$ can be further improved to $\log n$, or it is already tight. In the other direction, it is also worth questioning whether the total FPT inapproximability of {\sc $k$-SetCover} can be based on the weaker assumption {\sf W[2]$\neq$FPT}.

Current state-of-the-art inapproximability results of {\sc $k$-SetCover} are reached by two contrasting methods, namely, \textit{Distributed PCP Framework} \cite{KLM19} and \textit{Threshold Graph Composition} \cite{Lin19}, which we will introduce in Section \ref{DPCPF} and \ref{SetCover-TGC}, respectively. See the following table for an overview of results.

\begin{center}
\begin{tabular}{c|c|c|c}
Complexity Assumption & Inapproximability Ratio & Time Lower Bound   & Reference   \\ \hline
\multirow{2}{*}{{\sf W[1]$\neq$FPT}} & $(\log n)^{\varepsilon(k)=o(1)}$ & \multirow{2}{*}{$f(k)\cdot \text{poly}(n)$} & \cite{Lin19} \\
	& $(\log n)^{1/\text{poly}(k)}$ & & \cite{KLM19} \\ \hline
\multirow{2}{*}{{\sf ETH}} & $\left(\frac{\log n}{\log \log n}\right)^{1/k}$ & \multirow{2}{*}{$f(k)\cdot n^{\Omega(k)}$} & \cite{Lin19} \\
	& $(\log n)^{1/\text{poly}(k)}$ & & \cite{KLM19} \\ \hline
\multirow{2}{*}{{\sf SETH}} & $\left(\frac{\log n}{\log \log n}\right)^{1/k}$ & \multirow{2}{*}{$f(k)\cdot n^{k-\varepsilon}$} & \cite{Lin19} \\
	& $(\log n)^{1/\text{poly}(k)}$ & & \cite{KLM19} \\ \hline
\multirow{2}{*}{{\sf $k$-SUM Hypothesis}} & $\left(\frac{\log n}{\log \log n}\right)^{1/k}$ & \multirow{2}{*}{$f(k)\cdot n^{\lceil k/2\rceil-\varepsilon}$} & \cite{Lin19} \\
	& $(\log n)^{1/\text{poly}(k)}$ & & \cite{KLM19} \\ \hline
\end{tabular}
\end{center}

As for the coloring, the colored version and uncolored version of (exact) {\sc $k$-SetCover} are also equivalent, since we can add $k$ elements in the universe to ensure there is a set from each group, to reduce a colored version to an uncolored version. In the other direction, taking $k$ copies of the sets also works, because choosing replicated sets does not contribute. In the approximating sense, since it is a minimization problem, in the soundness case when solutions of size $\le g(k) \cdot k$ are ruled out, it always means we can't find such many sets to cover the universe even when choosing from the same set is allowed. Thus there is no need to specify the coloring.

\subsection{Hypercube Partition System}
We first introduce hypercube partition system, which is a powerful tool used in the reduction from {\sc MinLabel} to {\sc $k$-SetCover} \cite{CCK+17,KLM19}, and in creating the gap instance of {\sc $k$-SetCover} \cite{Lin19}.

\begin{defn}[Hypercube Partition System]
	The $(\kappa,\rho)$-hypercube partition system consists of the universe $\mathcal M$ and a collection of subset $\{P_{x,y}\}_{x \in [\rho],y \in [\kappa]}$ where $\mathcal M=[\kappa]^\rho$ and $P_{x,y}=\{z \in \mathcal M:z_x=y\}$.
\end{defn}

The universe $\mathcal M$ consists of all functions from $[\rho]$ to $[\kappa]$, and is of size $\kappa^\rho$. Each subset $P_{x,y}$ ($x \in [\rho],y \in [\kappa]$) consists of all functions mapping $x$ to $y$. It can be observed that one can cover the universe by picking all $\kappa$ subsets from some row $x \in [\rho]$, and this is the only way to cover the universe. In other words, even if we include $\kappa-1$ subsets from every row, it is not possible to cover the universe.

\subsection{Reduction from {\sc MinLabel}}
\label{MinLabel_to_SetCover}

Now we show how to reduce {\sc MinLabel} to {\sc SetCover}. This reduction preserves gap, but significantly increases the instance size.

Given a {\sc MinLabel} instance $\Gamma$ with $\ell$ left super-nodes $U_1 , \ldots , U_\ell$ and $h$ right super-nodes $W_1, \ldots, W_h$, we build a {\sc SetCover} instance as follows. Take $\ell$ different copies of $(h,|\Sigma_U|)$-hypercube partition system and set the universe to be the union of $\ell$ universes. Each set in {\sc SetCover} corresponds to a right vertex in {\sc MinLabel}. For a set $S_v$ associated to a right vertex $v \in W_j$ and for a left vertex $u \in U_i$, if there is an edge $(u,v)$, we include $P_{u,j}$ in set $S_v$. 

In order to see there is a one-to-one mapping from a solution of $\Gamma$ to a solution of the new {\sc SetCover} instance, note that for a left vertex $u \in U_i$, if a right multi-labeling covers $u$, by picking corresponding sets we have $P_{u,j}$ for all $j \in [h]$. Moreover, the only way to cover the universe is to include all $P_{u,j}$ for some row indexed by $u$, so a valid {\sc SetCover} solution must contain  sets associated to at least one neighbor in each right super-nodes for some specific left vertex $u$. The same argument applies to each of the $\ell$ left super-nodes, because we have a different copy of the hypercube partition system for each of them.

One important thing about this reduction is that the instance size is blowed up to $\ell \cdot h^{|\Sigma_U|}$, where $h$ is the solution size of yes instance (and also the parameter of {\sc $k$-SetCover}). Thus, in order to make it $\le f(k)n^{O(1)}$, the left alphabet size can not exceed $\frac{\log n}{\log k}$. Following the hardness of {\sc MinLabel} (Theorem \ref{MinLabel}) and letting $1/\gamma=(\log n)^{O(1)}$, {\sc $k$-SetCover} is hard to approximate to a $(\log n)^{O(\frac{1}{k})}$ factor assuming {\sf Gap-ETH}.

\subsection{Gap Producing via Distributed PCP}
\label{DPCPF}

In this section we introduce the \textit{Distributed PCP Framework}. This framework was first proposed by Abbound et al. \cite{ARW17}, and later used by Karthik et al. to rule out FPT approximation algorithms for {\sc $k$-SetCover} \cite{KLM19}. The interesting part of \cite{KLM19} is to obtain hardness of {\sc MaxCover} with specific parameters, while hardness of {\sc MinLabel} and {\sc $k$-SetCover} directly follow from reductions in \cite{CCK+17} (see Theorem \ref{MinLabel} and Section \ref{MinLabel_to_SetCover} respectively). 

At a high level, in this framework, one first rewrites the problem related to the hypothesis as a communication problem, then derives a \textit{Simultaneous Message Protocol} for this problem and extracts an instance of {\sc MaxCover} from the transcript of the protocol. 

Due to space limitations, we only introduce their {\sf W[1]} and {\sf ETH} results to give an overview of their methods. The ideas in {\sf SETH} and {\sf $k$-SUM Hypothesis} results are similar, except that they involve more complicated error correcting codes and protocols.

The similarity between {\sf W[1]$\neq$FPT} and {\sf ETH} is that they both often involve agreement tests (in other words, equality tests). Starting from {\sf W[1]$\neq$FPT}, one may want to set $K=\binom{k}{2}$ groups, each containing $N=|E|$ elements representing valid edges. The goal is to pick an edge from each group such that the label of end points (which is of $\log N$ bits length) are consistent. {\sf W[1]$\neq$FPT} states that this problem does not admit an $f(K)N^{O(1)}$ algorithm. Similarly, starting from {\sf ETH}, one may also want to divide the clauses into $K$ groups, each containing at most $N=2^{\Theta(n/K)}$ partial satisfying assignments for those clauses. The goal is also to pick an assignment from each group such that the values on each variables (which is of $n=\Theta(K \log N)$ bits length in total) are consistent. {\sf ETH} states that this cannot be done in $N^{o(K)}$ time.

We can think of the agreement test problem as a communication problem: there are $K$ players, each receiving an element from the corresponding group. They want to collaborate nicely to decide whether their elements in hand ``agree'' or not. To achieve this, they use a specific communication protocol called \textit{Simultaneous Message Protocol}, which was introduced by Yao \cite{Yao79}:

\begin{defn}[Simultaneous Message Protocol]
	We say $\pi$ is a $(r,\ell,s)$-efficient protocol if the following holds:
	\begin{itemize}
		\item The protocol is one-round with public randomness. The following actions happen sequentially:
			\begin{enumerate}
				\item The players receive their inputs.
				\item The players and the referee jointly toss $r$ random coins.
				\item Each player on seeing the randomness deterministically sends an $\ell$-bit message to the referee.
				\item Based on the randomness and the $K\cdot \ell$ bits sent from the players, the referee outputs accept or reject.
			\end{enumerate}
		\item The protocol has completeness 1 and soundness $s$, i.e.,
			\begin{itemize}
				\item If their inputs indeed agree, then the referee always accepts regardless of randomness.
				\item Otherwise, the referee accepts with probability at most $s$.
			\end{itemize}
	\end{itemize}
\end{defn}

The full version of SMP protocol also admits $w$ bits of advice. In the contexts of {\sf W[1]$\neq$FPT} and {\sf ETH} here, we do not need this.

Note that by repeating the protocol $z$ times, each time using fresh randomness, we can derive a $(z\cdot r,z \cdot \ell,s^z)$-efficient protocol from an $(r,\ell,s)$-efficient protocol. 

Next we will see how to use {\sc MaxCover} to simulate an SMP protocol. Then the hardness of the starting problem ({\sc $k$-Clique} or {\sc 3SAT}) and the existence of an efficient protocol for it will lead to hardness of {\sc MaxCover}. 

\begin{thm}[Theorem 5.2 in \cite{KLM19}, slightly simplified]
	An instance $\Pi$ of a ({\sc $k$-Clique} or {\sc 3SAT}) problem which admits an $(r,\ell,s)$-efficient $K$-player SMP protocol can be reduced to a {\sc MaxCover} instance $\Gamma$ as follows:
	\begin{itemize}
		\item The reduction runs in time $2^{r+K\cdot \ell}\cdot \text{poly}(N,K)$.
		\item $\Gamma$ has $K$ right super-nodes of size at most $N$ each.
		\item $\Gamma$ has $2^{r}$ left super-nodes of size at most $2^{K \cdot \ell}$ each.
		\item If $\Pi$ is a YES instance, then {\sc MaxCover}$=1$.
		\item If $\Pi$ is a NO instasnce, then {\sc MaxCover}$\le s$.
	\end{itemize}
\end{thm}
\begin{proof}
	Here is the construction: 
	\begin{itemize}
		\item Each right super-node corresponds to the group which a player receives an input from.
		\item Each left super-node corresponds to a random string $\gamma\in \{0,1\}^r$. The left super-node $U_\gamma$ contains one node for each  possible accepting messages from the $K$ players, i.e., each vertex in $U_\gamma$ corresponds to $(m_1,\ldots,m_K)\in(\{0,1\}^\ell)^K$ where in the protocol the referee accepts on seeing randomness $\gamma$ and messages $(m_1,\ldots,m_K)$.
		\item We add an edge between a right vertex $x \in W_j$ and a left vertex $(m_1,\ldots,m_K)\in U_\gamma$ if $m_j$ is equal to the message that player $j$ sends on an input $x$ and randomness $\gamma$ in the protocol.
	\end{itemize}
	Detailed proofs of the desired properties are omitted here since they are rather straightforward.
\end{proof}

In the last step, we need to derive an efficient SMP protocol for {\sc $k$-Clique} and {\sc 3SAT}. Directly sending their inputs (which is of $\ge\log N$ bits length) does not seem to be a good idea because it would make the size of {\sc MaxCover} to be $2^{r+K \cdot \ell}=2^{\Omega(K)\log N}=N^{\Omega(K)}$, which is too large. What's more, since the further reduction from {\sc MaxCover} to {\sc MinLabel} then to {\sc $k$-SetCover} will introduce a $K^{|\Sigma_U|}$ blow-up, we even need $|\Sigma_U|=2^{K \cdot \ell} \le \frac{\log N}{\log K}$.

Fortunately, this can be done via a simple error correcting code called \textit{good code} which has constant rate and constant relative distance. To check the consistency of their inputs, we only need to check the equality of one random bit of their codes. The randomness is used to specify the index of that bit. This way we can have $r=O(\log \log N), \ell=O(1), s=O(1)$. Within the bound that $\ell \le \frac{1}{K}(\log \log N-\log \log K)$, we can repeat the protocol $z=O(\frac{1}{K} \log \log N)$ times, leading to a soundness parameter $s=O(1)^z=O((\log N)^{\frac{1}{K}})$. After the reduction to {\sc MinLabel}, the $\varepsilon$ in gap would become $\varepsilon^{\frac{1}{K}}$.

Along such a long way we finally reach a $(\log n)^{\frac{1}{\text{poly}(k)}}$ inapproximability ratio for {\sc $k$-SetCover} with FPT time lower bound under {\sf W[1]$\neq$FPT} and $n^{o(k)}$ time lower bound under {\sf ETH}.

\subsection{Gap Producing via Threshold Graph Composition}
\label{SetCover-TGC}

In this subsection we introduce how to obtain $\left(\frac{\log n}{\log \log n}\right)^{1/k}$ inapproximability of {\sc $k$-SetCover} via threshold graph composition technique.

In general, this technique transforms an {\sc $k$-SetCover} instance with small universe (typically $|U|=\Theta(\log n)$) still to an instance of {\sc $k$-SetCover}, increasing the size of universe to $O(n)$ while creating a gap. In the YES case, the number of sets needed to cover the universe is still $k$, but in the NO case it becomes $h \gg k$, where $h$ is determined by the threshold graph.

Given an {\sc $k$-SetCover} instance $\Gamma=(\mathcal S,U)$, we need a bipartite threshold graph $G=(A \dot\cup B,E)$ with the following properties:
\begin{enumerate}
	\item $A$ is not divided. $B$ is divided into $\ell$ groups: $B=B_1 \dot\cup\ldots\dot\cup B_\ell$, where $\ell$ is arbitrary.
	\item $|A|=n, |B_i| \le \frac{\log n}{\log |U|}, \forall i \in [\ell]$.
	\item For any $k$ vertices $a_1,\ldots, a_k \in A$ and for any $i \in [\ell]$, there is at least one common neighbor of $\{a_1,\ldots,a_k\}$ in $B_i$.
	\item For any $X \subseteq A$, if for any $i \in [\ell]$, there is at least one vertex in $B_i$, which is a common neighbor of at least $k+1$ vertices in $X$, then $|X|\ge h$.
\end{enumerate}

We compose the original {\sc $k$-SetCover} instance $\Gamma=(\mathcal S,U)$ with this threshold graph $G$ to produce an instance $\Gamma'=(\mathcal S',U')$ as follows. 

\begin{itemize}
	\item $|\mathcal S'|=|\mathcal S|$. For each $S \in \mathcal S$ we associate a new set $S'\in \mathcal S'$. We also associate a vertex in $A$ (the left side of the threshold graph) to each set $S' \in \mathcal S'$. 
	\item $\Gamma$ consists of $\ell$ hypercube partition systems, one for each $B_i$ (the right side of the threshold graph). The $i$-th partition system has $|B_i|$ rows and $|U|$ columns, thus is of size $|U|^{|B_i|}$.
	\item For any $i \in [\ell], x \in B_i, y \in U$, subset $P_{x,y}$ is included in a set $S' \in \mathcal S'$ if and only if:
		\begin{enumerate}
			\item $y \in S$, i.e., set $S$ can cover $y$ in the original instance $\Gamma$.
			\item There is an edge between the vertex associated to $S'$ and vertex $x$ in the threshold graph $G$.
		\end{enumerate}
\end{itemize}

As shown above, each row in the partition system corresponds to a vertex in $B_i$, and each column corresponds to an element in $U$. According to the properties of partition system, in order to cover the universe, we must pick all $|U|$ subsets in some specific row. This, together with our construction, means there is a vertex $x \in B_i$ such that sets correspond to its neighbors in $G$ can cover $U$. Since the $\ell$ hypercube partition systems are independent, this holds for each $B_i, i \in [\ell]$.

In the YES case, $k$ sets are enough to cover $U$, and there is at least one common neighbor of them in every $B_i, i \in [\ell]$. Thus the answer to the new instance is still $k$.

In the NO case, at least $k+1$ sets are need to cover $U$. Consider any solution $X \subseteq \mathcal S'$ of $\Gamma'$, we know that in each group $B_i$, there is a vertex $x \in B_i$ such that $|\mathcal N(x) \cap X| \ge k+1$ (because they form a solution in $\Gamma$). Therefore, by property 4 of threshold graph, $|X| \ge h$ as desired.

The last step is to construct a threshold graph the desired properties. \cite{Lin19} used a specific combinatorial object called universal sets to construct it. However, the graph in Section \ref{TGC-KN21} with proper parameters also suffices, and is simpler. 

The required property is closely related to the collision property in the construction of Section \ref{TGC-KN21}. In fact, the gap $h$ here is just the collision number $Col(C)$ of the error correcting code. Thus, we want codes with large collision numbers.

Let the alphabet of the code be $\Sigma$, then it's easy to see the collision number cannot be larger than $|\Sigma|+1$, since such many strings must collide in every position by pigeon principle. However, this upper bound can be reached by codes constructed from \textit{perfect hash family}.

\begin{defn}[Perfect Hash Family]
	For every $N,\ell\in \mathbb N$ and $\Sigma$, we say that $H:=\{h_i:[N] \to \Sigma | i \in [\ell]\}$ is a $[N,\ell]_{|\Sigma|}$-perfect hash family if for every subset $T\subseteq [N]$ of size $|\Sigma|$, there exists an $i\in [\ell]$ such that
	$$\forall x ,y \in T, x \ne y, h_i(x) \ne h_i(y)$$
\end{defn}

Think of an $[N,\ell]_{|\Sigma|}$-perfect hash family as an $N \times \ell$ matrix, where each column represents a hash function. Then the property is that for any $\le |\Sigma|$ rows, there is a column with no collisions (on these rows). Regard each row as a codeword of an error correcting code from $\Sigma^{\log_{|\Sigma|} N} \to \Sigma^\ell$, then the collision number of this error correcting code is $>|\Sigma|$, thus exactly $|\Sigma|+1$ as desired.

\begin{lema}[Alon et al. \cite{AYZ95} ]
	For every $N,|\Sigma| \in \mathbb N$ there is a $[N,2^{O(|\Sigma|)}\cdot \log N]_{|\Sigma|}$-perfect hash family that can be computed in $\tilde O_{|\Sigma|}(N)$ time.
\end{lema}

In this threshold graph the size of a right super-node $B_i$ is $|\Sigma|^k$. Assuming $|U|=O(\log n)$, to make $|B_i| \le \frac{\log n}{\log |U|}=\frac{\log n}{\log \log n}$, we need $|\Sigma| \le \left(\frac{\log n}{\log \log n}\right)^{1/k}$ and this is the best gap possible. By the above lemma, the perfect hash family (and the corresponding code) can be constructed efficiently.

Our last step is to show that, assuming different hypothesis, {\sc $k$-SetCover} with universe size $\le \log n$ is still hard. The reduction from {\sc 3SAT} is straightforward: divide the variables into $k$ groups of size $n/k$ each, a group of sets corresponds to $N=2^{n/k}$ possible assignments to those $n/k$ variables, and the universe is just the $m=\Theta(k \log N)$ clauses. {\sf ETH} (respectively, {\sf SETH}) asserts that this instance cannot be solved in $N^{o(k)}$ (respectively, $N^{k-\varepsilon}$) time. Thus by the above threshold graph composition technique, we reach $\left(\frac{\log n}{\log \log n}\right)^{1/k}$-inapproximability of {\sc $k$-SetCover} based on {\sf ETH} and {\sf SETH}.

The reduction from {\sc $k$-Clique} to {\sc $k$-SetCover} is a bit more complicated, as stated in the following theorem. The main idea is from Karthik et al. \cite{KLM19}.
\begin{thm}
	There is an $\text{poly}(n)$ time algorithm, which can reduce an {\sc $k$-Clique} instance to an {\sc $\binom{k}{2}$-SetCover} instance with universe size $\text{poly}(k) \log n$.
\end{thm}
\begin{proof}
	Sets in a group $(i,j)$ still represent edges between the $i$-th block and the $j$-th block (in the {\sc $k$-Clique} instance). In order to check the consistency of labels, we need $k \times \log n$ hypercube partition systems, one for each $(i,\ell), i \in [k],\ell \in [\log n]$. The $(i,\ell)$-th partition system is meant to check whether the $\ell$-th bits of labels of vertices in block $i$ are all the same. In an invalid solution, one may pick an edge between the $i$-th and $j$-th blocks, and an edge between the $i$-th and $j'$-th blocks, such that the two vertices (let them be $v_1$ and $v_2$) in the $i$-th block are not the same. In such a case, $v_1$ and $v_2$ must differ in at least one bit, and thus cannot fulfill the requirements in the $(i,\ell)$-th partition system where $\ell$ is the position of that bit.
	
	Specifically, for all $i \in [k],\ell \in [\log n]$, the $(i,\ell)$-th partition system contains 2 rows and $(k-1)$ columns. The choices of rows represent the choices of the bit to be 0 or 1, and the columns test agreement of the $(k-1)$ labels (edges between the $i$-th block and each remaining blocks). For an edge between $v \in V_i$ and $w \in V_j$, we include $P_{v[\ell],j}$ into its set. At last, the {\sc $\binom{k}{2}$-SetCover} instance is the union of those $k \cdot \log n$ hypercube partition systems.
	
	The instance size is $\text{poly}(n)$ since there are that many edges, while the universe size is $k \cdot \log n\cdot (k-1)^2=\text{poly}(k)\cdot \log n$.
\end{proof}

The {\sf W[1]}-hardness of {\sc $k$-SetCover} then follows. It is worth noting that in such FPT reductions, the parameter $k'$ can be arbitrarily amplified as long as it is a function of $k$. Assuming {\sf W[1]$\neq$FPT}, {\sc $k$-SetCover} is hard to approximate to within a factor of $\left(\frac{\log n}{\log \log n}\right)^{1/\binom{k}{2}}$, then for any function $\varepsilon(k)$ which is $o(1)$ as $k$ goes to infinity, take large enough $k'$ such that $\varepsilon(k')<1/\binom{k}{2}$, then for large enough $n$ we have $(\log n)^{\varepsilon(k')}<\left(\frac{\log n}{\log \log n}\right)^{1/\binom{k}{2}}$, which means {\sc $k'$-SetCover} (by padding the parameter from $k$ to $k'$) cannot be approximated to a factor of $(\log n)^{\varepsilon(k')}$ in FPT time.

Instead of introducing their {\sf $k$-SUM Hypothesis} result here, we want to make some comments on this technique. Note that the maximum size of a right super-node in the threshold graph is $\frac{\log n}{\log |U|}$. Thus when $|U|$ is not as small as $\log n$, it may still be possible to obtain some inapproximability results. It remains a big open question that whether we can base total FPT inapproximability of {\sc $k$-SetCover} on {\sf W[2]$\neq$FPT}. 
If we can construct threshold graphs with a gap such that each right super-nodes consist of only $O(1)$ vertices, we can obtain {\sf W[2]} hardness of {\sc $k$-SetCover}, respectively. Note that the construction in Section \ref{TGC-KN21} does not suffice, because the size of their right super-nodes is $\Sigma^k$, which is too large even if $|\Sigma|=O(1)$.